\documentclass[useAMS,usenatbib,letterpaper]{biom}
\usepackage{amsfonts}
\usepackage{graphicx}
\usepackage{amsmath}
\usepackage{bm}
\usepackage{booktabs}
\usepackage{multirow}
\usepackage{caption}
\usepackage{tikz}
\usepackage{subcaption}
\usepackage{physics}

\def\*#1{\bm{#1}}
\def\bSig\mathbf{\Sigma}

\renewcommand{\vec}[1]{\mathbf{#1}}

\title{Novel Non-Negative Variance Estimator for (Modified) Within-Cluster Resampling}

\author{Daniel Xu$^*$\email{danielxu@wharton.upenn.edu}, Pamela Shaw$^{**}$\email{shawp@pennmedicine.upenn.edu}, and Ian Barnett$^{***}$\email{ibarnett@pennmedicine.upenn.edu} \\
Department of Biostatistics, Epidemiology, and Informatics \\
 University of Pennsylvania, Philadelphia, PA, U.S.A.}

\begin{document}


\date{{\it }  {\it  }. {\it
 } }



\pagerange{\pageref{firstpage}--\pageref{lastpage}} 
\volume{}
\pubyear{}
\artmonth{}


\doi{XX}


\label{firstpage}


\begin{abstract}
This article proposes a novel variance estimator for within-cluster resampling (WCR) and modified within-cluster resampling (MWCR) -- two existing methods for analyzing longitudinal data. WCR is a simple but computationally intensive method, in which a single observation is randomly sampled from each cluster to form a new dataset. This process is repeated numerous times, and in each resampled dataset (or outputation), we calculate $\beta$ using a generalized linear model. The final resulting estimator is an average across estimates from all outputations. MWCR is an extension of WCR that can account for the within-cluster correlation of the dataset; consequently, there are two noteworthy differences: 1) in MWCR, each resampled dataset is formed by randomly sampling multiple observations without replacement from each cluster and 2) generalized estimating equations (GEEs) are used to estimate the parameter of interest. While WCR and MWCR are relatively simple to implement, a key challenge is that the proposed moment-based estimator is often times negative in practice. Our modified variance estimator is not only strictly positive, but simulations show that it preserves the type I error and allows statistical power gains associated with MWCR to be realized.
\end{abstract}

%



\maketitle


%

\section{Introduction}
\label{sec:intro}

Generalized estimating equations (GEEs) have become one of the most popular marginal methods for analyzing longitudinal data \citep{liang1986longitudinal}. Through the use of sandwich variance estimators and working correlation matrices, GEEs provide a semiparametric approach that accounts for all observations per subject and is robust to model misspecification \citep{huber1967behavior, eicker1963asymptotic}. However, despite the power and popularity of GEEs, they are not suitable for all cases. One such case occurs when cluster sizes are large compared to the number of individuals, which is typically the case in digital phenotyping studies in which a few individuals are followed for long periods of time \citep{barnett2018beyond}. In these instances, the sandwich estimator becomes numerically unstable, and a regularized sandwich estimator may be used instead in order to shrink the sample correlation matrix towards the working correlation matrix \citep{warton2011regularized}. Another case in which GEEs fail is when cluster sizes are informative, which can lead to biased estimators under the standard GEE approach \citep{wang2011inference}. This concern has led to the development of two modified approaches that are asymptotically equivalent -- a Monte Carlo within-cluster sampling approach and a cluster-weighted GEE approach \citep{hoffman2001within, williamson2003marginal}.

Multiple outputation (MO) is the term used by Follmann et al. to describe the within-cluster resampling approach (WCR) by Hoffman et al., which allows us to perform statistical procedures typically requiring independent data on clustered data. The MO estimate of the parameter of interest is the average of the parameter across multiple outputations, where in each outputation, the parameter is estimated from a dataset generated by randomly selecting a single observation per individual \citep{follmann2003multiple}. Though computationally intensive, MO provides a simple approach to analyzing clustered data and consequently has been used in genetic association studies as well as clinical studies \citep{askie2010inhaled, tian2007comparing}. The key difference between MO and WCR is that WCR is often used in the context of generalized linear models, whereas MO generally has a broader scope.

However, the key drawback of MO or WCR is that in each outputation, only a single observation is randomly sampled per subject and the rest of the data is discarded. Therefore, we effectively ignore the within-subject correlation, which may contain information to improve statistical power. \citet{chiang2008efficient} specifically address this issue -- their approach is known as modified within-cluster resampling (MWCR), and in each outputation, GEEs are used with multiple observations per subject to calculate the parameter of interest.

Despite these advances, there are still certain limitations of MWCR, which are largely associated with the fact that MWCR has been studied primarily in the large sample, asymptotic limit. We have identified that a key obstacle to successfully implementing MWCR is that while the moment-based variance estimator proposed by Hoffman is easy to compute, it is not strictly positive. This can lead to significant issues when the number of outputations is not large enough.

In this paper, we present a nonnegative, stabilized variance estimator for WCR and MWCR. To present our methodology, we first develop the appropriate terminology for WCR and MWCR before presenting our new variance estimator in Section 2. In Section 3, we apply our methodology to a simulated dataset. We find that our new variance estimator largely preserves the type I error rate and allows increases in statistical power to be realized when the correct working correlation structure is chosen for MWCR. Finally, in Section 4, we compare the performance of our new variance estimator to more traditional methods with a longitudinal dataset of schizophrenia patients.

\section{Model and notation}
\begin{figure*}[h]
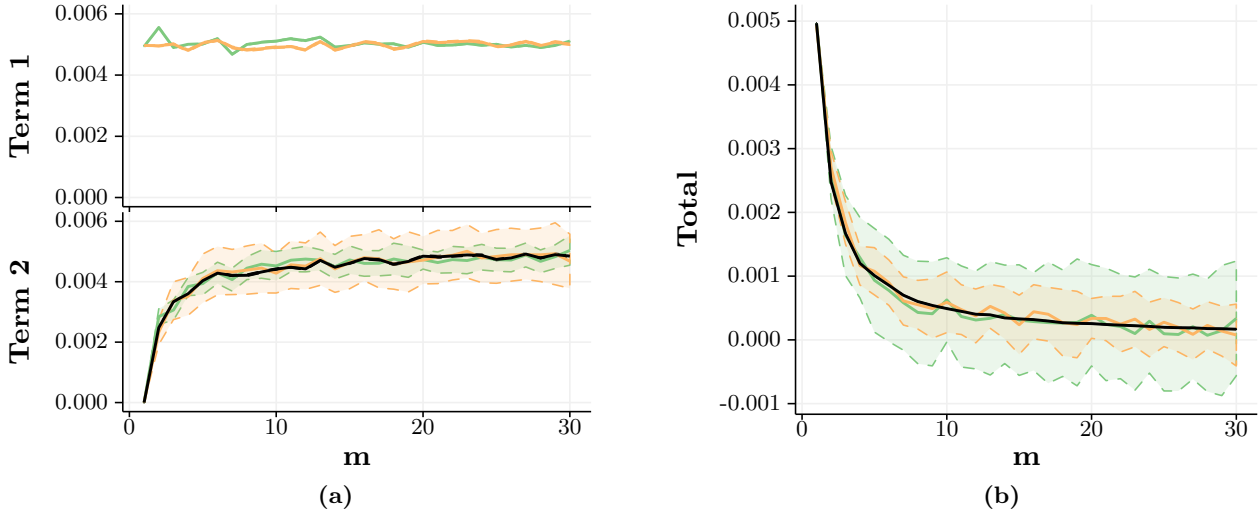

  \centering
  \begin{subfigure}{.5\textwidth}
    \input{term1and2.tex}
    \caption{}
    \label{fig:negvarind}
  \end{subfigure}%
  \begin{subfigure}{.5\textwidth}
    \input{totalvar.tex}
    \caption{}
    \label{fig:negvartotal}
  \end{subfigure}%
  \caption{Illustration of performance of our stabilized variance estimator compared to the traditional moment-based approach. Data was generated with an "exchangeable" correlation structure via a random intercepts model with $n = 1000$, and consequently analyzed with MO-2 with an assumed "independence" working correlation structure, where m was varied from 2 to 30. The orange line represents the traditional moment-based estimator with $M = 250$, the green line represents the traditional moment-based estimator with $M = 1000$, and the black line represents our stabilized variance estimator with $M = 250$. In (a), we plot the first ($\bm{\bar{\hat{\Sigma}}}$) and second term ($\vec{S}^2_{\bm{\hat{\beta}} | \vec{\underline{X}}}$) of Equation \ref{eq:momentbasevar}. The variability associated with the second term in particular is evident given the finite number of outputations. In (b), the total variance (i.e. $\bm{\bar{\hat{\Sigma}}} - \vec{S}^2_{\bm{\hat{\beta}} | \vec{\underline{X}}}$) is plotted. In all figures, confidence bands represent $\pm2$ SDs associated with 50 repeats.}
  \label{fig:negvar}
\end{figure*}
\subsection{Modified within-cluster resampling}
\label{sec:Notation}
Given a dataset containing $n$ independent subjects and $m_i$ observations for subject $i$, we can denote the data by the long row vectors $\vec{\underline{Y}} = (Y_{11},...,Y_{1m_1},...,Y_{21},...,Y_{2m_2},...,Y_{n1},...,Y_{nm_n})$ and $\vec{\underline{X}} = (\vec{X}_{11},...,\vec{X}_{1m_1},...,\vec{X}_{21},...,\vec{X}_{2m_2},...,\vec{X}_{n1},...,\vec{X}_{nm_n}),$ where $Y_{ij}$ is the outcome and $\vec{X}_{ij}$ is a $p \times 1$ covariate vector for the $j$th observation for person $i$. Because the number of observations may vary across individuals, we denote $m$ as $min(m_i)$. In the MWCR approach described by \citet{chiang2008efficient}, we would first randomly sample $B$ observations per person, where $2 \leq B \leq m$. Consequently, we let the long row vector $\vec{\underline{J}} = (\vec{J}_{1},\vec{J}_{2},...,\vec{J}_{n})$ signify the indices of the observations chosen in each outputation, where $\vec{J}_{i}$ represents the indices of the observations chosen for person $i$.

From here, following the approach by \citet{liang1986longitudinal}, for a single outputation, we let $\bm{\hat{\beta}}\{\vec{X(J)}\}$ represent the solution to the generalized estimating equation
\begin{equation*}
    \sum_{i=1}^{n} \vec{D}^T_i \vec{V}^{-1}_i (\vec{Y}_i-\bm{\mu}_i) = 0
\end{equation*}
where $\vec{D_i}=\frac{\delta \bm{\mu}_i}{\delta \bm{\beta}^T}$ and $\vec{V}_i$ is the $n_i \times n_i$ working covariance matrix. More specifically, let $\vec{V}_i$ be defined as
\begin{equation*}
    \vec{V}_i=\vec{V}^{1/2}_{M_i}\vec{R}_{i}(\bm{\alpha})\vec{V}^{1/2}_{M_i}
\end{equation*}
where $\vec{V}_{M_i}=diag\{\phi a^{-1}_{ij} v(\mu_{ij}) \}$ is the marginal variance of $\vec{Y}_i$ and $\vec{R}_i(\bm{\alpha})$ is the working correlation matrix.

By repeating this random sampling process for every possible combination of $N$ observations, we could calculate all values of $\bm{\hat{\beta}}$ given $\bm{\underline{X}}$. Thus, it follows that $E[\bm{\hat{\beta}}\{\vec{X(J)}\} | \bm{\underline{X}}]$ represents the expected multiple outputation (EMO) estimate
\begin{equation*}
    \bm{\bar{\hat{\beta}}}^{\infty} = \frac{\sum_{\vec{J}}\bm{\hat{\beta}}\{\vec{X(J)}\}}{\prod \binom{m_i}{B}}.
\end{equation*}

Often times, however, there are too many samples to compute, so we define the multiple outputation (MO) estimate of $\bm{\beta}$ from $M$ outputations to be
\begin{equation*}
    \bm{\bar{\hat{\beta}}}^{M} = \sum_{k=1}^{M}\frac{\bm{\hat{\beta}}\{\vec{X(J)}\}}{M}.
\end{equation*}

Following the variance decomposition approach by \citet{hoffman2001within} and \citet{follmann2003multiple}, we can see that the variance for our estimate is
\begin{equation}
\begin{split}
    \hat{var}_{MO} (\bm{\bar{\hat{\beta}}}^{M}) &= \sum_{k = 1}^{M} \frac{\bm{\hat{\Sigma}}_{k}}{M}- \sum_{k = 1}^{M} \frac{(\bm{\hat{\beta}}_{k} - \bm{\bar{\hat{\beta}}}^{M})^2}{M-1}\\
    &= \bm{\bar{\hat{\Sigma}}} - \vec{S}^2_{\bm{\hat{\beta}} | \vec{\underline{X}}},
\end{split}
\label{eq:momentbasevar}
\end{equation}
where $\bm{\hat{\Sigma}}_{k}$ represents the covariance matrix of the $kth$ outputation and $\vec{S}^2_{\bm{\hat{\beta}} | \vec{\underline{X}}}$ represents the sample covariance of the $\bm{\hat{\beta}}$'s. As a final consideration, because $B$ can take on a range of values from 2 to $m$, we refer to this method collectively as MO-B.

Throughout this paper, if $B = 1$, we refer simply to the MO scenario. This describes the WCR approach proposed by \citet{hoffman2001within}, where $\bm{\hat{\beta}}\{\vec{X(J)}\}$ for each outputation is the maximum likelihood estimator for a generalized linear model based on $n$ independent observations. Everything else -- including the variance for the MO estimate -- remains the same regardless of $B$.


\subsection{Negative variance estimator}
Though the moment-based variance estimator in Equation \ref{eq:momentbasevar} is straightforward to calculate, in practice, it can lead to negative values a significant portion of the time, which is illustrated in Figure \ref{fig:negvartotal}. This problem can be mitigated by choosing a substantially large enough M, but as \cite{follmann2003multiple} suggest, even in the simulated data we used to generate Figure \ref{fig:negvar} when $m = 30$, we would need to choose $M > 625,000$ to ensure
\begin{equation}
    Pr(\abs{Z^{M}-Z^{\infty}}<.02|\bm{\underline{X}})=.95,
\label{eq:Mcalc}
\end{equation}
 where $Z^M = \bm{\bar{\hat{\beta}}}^{M}/\sqrt{\hat{var}_{MO} (\bm{\bar{\hat{\beta}}}^{M})}$. As the number of observations per person grows and $\vec{S}^2_{\bm{\hat{\beta}} | \vec{\underline{X}}} \to \bm{\bar{\hat{\Sigma}}}$, the number of outputations needed to satisfy Equation \ref{eq:Mcalc} increases nontrivially.

\begin{theorem}\label{thm1}
If $\bm{\hat{\beta}}$ represents the OLS estimate, the expectation of the second term in equation \eqref{eq:momentbasevar} with $M=\infty$ is
\begin{equation*}
    E[\vec{S}^2_{\bm{\hat{\beta}} | \vec{\underline{X}}}] = \bm{\bar{\hat{\Sigma}}} \left(1 - \frac{1}{n}\sum_{i = 1}^{n}\frac{B}{m_i}\right).
\end{equation*}
\end{theorem}
\begin{proof} First we define $J_k$, $\bm{\beta}_k$, $\bm{\beta}_{k-l}$, and $\bm{\beta}_{kl}$.

$\begin{array}{rl}
E[\vec{S}^2_{\bm{\hat{\beta}} | \vec{\underline{X}}}] =& \bm{\Sigma}-\frac{1}{{m \choose B}^2} \sum_{J_j,J_l} E\left(\hat{\bm{\beta}}_j\hat{\bm{\beta}}_l^T\vert J_j,J_l\right)\\ 
&+ \frac{1}{{m \choose B}} \sum_{J_l} E\left(\hat{\bm{\beta}}_l \vert J_l\right)E\left(\hat{\bm{\beta}}_l^T \vert J_l\right) \\
=& \bm{\Sigma} - E\left\{E\left(\hat{\bm{\beta}}_j\hat{\bm{\beta}}_l^T \vert J_j,J_l\right)\right\} + \bm{\beta}\bm{\beta}^T \\
=& \bm{\Sigma} - \sum_{b=0}^B E\left(\hat{\bm{\beta}}_j\hat{\bm{\beta}}_l^T-\bm{\beta}\bm{\beta}^T \bigg\vert |J_j \cap J_l| = b \right) \\
& \cdot \mbox{pr}\left(|J_j \cap J_l| = b\right) \\
=& \bm{\Sigma} - \sum_{b=0}^BE\left(\left(X_j^TX_j\right)^{-1}\left(X_{jl}^TX_{jl}\right)\right. \\ 
& \left(\hat{\bm{\beta}}_{jl}\hat{\bm{\beta}}_{jl}^T-\bm{\beta}\bm{\beta}^T\right) \\
& \left.\left(X_{jl}^TX_{jl}\right)\left(X_l^TX_l\right)^{-1} \bigg\vert |J_j \cap J_l| = b\right)  \\
& \cdot \mbox{pr}\left(|J_j \cap J_l| = b\right) \\
=&\bm{\Sigma} - \sum_{b=0}^B \left(\frac{b}{B} \bm{I}_p\right) \left(\frac{B}{b} (\bm{\Sigma} + O(M^{-1}))\right)\left(\frac{b}{B} \bm{I}_p\right) \\
& \cdot \mbox{pr}\left(|J_j \cap J_l| = b\right) \\
=& \bm{\Sigma} \left\{1-\sum_{b=0}^B \frac{b}{B} \mbox{pr}\left(|J_j \cap J_l| = b\right) \right\} \\
=& \bm{\Sigma} \left(1-\frac{B}{m}\right) \\
\end{array}$
\end{proof}

Though \cite{follmann2003multiple} propose a positive variance estimate when $\bm{\hat{\beta}}$ is an MLE, we note that this only applies to us when $B = 1$. Therefore, we propose a new, strictly, positive variance estimator, which is seen in Theorem \ref{thm1}. We find that this new estimator is exact in the OLS case and empirically correct in the GEE case. Figure \ref{fig:negvar} demonstrates that this new variance estimator closely follows the expected variance and has practically no variability, even in the GEE case ($B = 2$) when our result is not exact and when we use a relatively small number of outputations ($M = 250$).

\section{Simulation study}
\begin{table*}[ht]
\centering
\resizebox{\textwidth}{!}{%
\begin{tabular}{p{1.4 cm} | p{1.4cm} | p{1.7cm} | p{1.7cm} | c | p{1.7cm} | p{1.7cm} |c | p{1.7cm} | p{1.7cm} | p{1.7cm} |c | p{1.7cm} | p{1.7cm} | p{1.7cm}}
  \toprule
  \\
  \\
  & & \multicolumn{5}{c}{(a) $n = 1000$} & & \multicolumn{7}{c}{(b) $n = 10$} \\
  \cmidrule{3-7} \cmidrule{9-15}
  & & \multicolumn{2}{c}{$m = 10$} & & \multicolumn{2}{c}{$m = 30$} & &
  \multicolumn{3}{c}{$m = 10$} & & \multicolumn{3}{c}{$m = 30$}\\
  \cmidrule{3-4} \cmidrule{6-7} \cmidrule{9-11} \cmidrule{13-15}
 & & Moment-based & Stabilized & & Moment-based & Stabilized & & Stabilized & BC Stabilized & BC Moment-based & & Stabilized & BC Stabilized & BC Moment-based\\ 
  \midrule
 $\alpha = .05$ & MO & .066 & .049 & & .077 & .051 & & .031 & -- & -- & & .024 & -- & -- \\ 
 & MO-2 & .055 & .051 & & .074 & .045 & & .086 & .023 & .001 & & .080 & .015 & .000 \\ 
 & MO-3 & .061 & .059 & & .065 & .053 & & .078 & .020 & .008 & & .066 & .021 & .000 \\ 
 & MO-4 & .052 & .050 & & .049 & .047 & & .067 & .023 & .015 & & .075 & .026 & .001 \\ 
 & MO-5 & .053 & .053 & & .055 & .053 & & .071 & .029 & .028 & & .074 & .020 & .002 \\[.3 cm]
 $\alpha = .01$ & MO & .020 & .009 & & .026 & .012 & & .006 & -- & -- & & .004 & -- & -- \\
 & MO-2 & .010 & .010 & & .026 & .011 & & .026 & .002 & .000 & & .017 & .002 & .000 \\ 
 & MO-3 & .008 & .009 & & .016 & .012 & & .014 & .004 & .000 & & .019 & .002 & .000 \\ 
 & MO-4 & .012 & .012 & & .011 & .007 & & .016 & .002 & .001 & & .017 & .004 & .000 \\ 
 & MO-5 & .007 & .006 & & .014 & .010 & & .022 & .009 & .007 & & .014 & .002 & .000 \\
   \bottomrule
\end{tabular}}
\caption{Comparison of type I error of MO and MO-B for $2 \leq B \leq 5$ at $\alpha = .01$ and $\alpha = .05$. (a) demonstrates the effect of varying $m$ under the large sample size setting ($n = 1000$ subjects). For each condition, 2,000 simulations with $M \geq 2750$, $\sigma = 1$, $\tau = 1$ were conducted under the null hypothesis of no effect, and in each simulation, both the moment-based and stabilized variance estimator were recorded. In the case of the moment-based estimator, negative variances (which comprised of less than .75\% of the simulations for each condition) were discarded. (b) demonstrates the effect of varying $m$ under the small sample size setting ($n = 10$ subjects) with the same conditions ($M \geq 2750$, $\sigma = 1$, $\tau = 1$). For each condition, the type I error associated with our stabilized variance estimator was recorded. For MO-B, the type I error associated with the bias-corrected (BC) stabilized and moment-based variance estimators proposed by Mancl and DeRouen were also recorded; the traditional moment-based estimator is not because of the higher prevalence of negative variances in these simulations.}
\label{table:typeIerror}
\end{table*}

\subsection{Type I error} 
\label{sec:typeIerror}
To validate our proposed variance estimator, we conducted 2,000 simulations of the null distribution for different values of $m$. Data for each of these simulations was generated following the underlying random intercepts model $Y_{ij} = \beta_0 + b_i + \beta_{1}X_i + \epsilon_{ij}$, where $b_i \sim \mathcal{N}(0, \tau)$ and $e_{ij} \sim \mathcal{N}(0, \sigma)$. For MO-B, data was analyzed using an "independence" working correlation structure. Proportion of significant p-values under the null hypothesis of no association for $\alpha = .05$ and $\alpha = .01$ are summarized in Table \ref{table:typeIerror}a. In general, we find that at both the $\alpha = .05$ and $\alpha = .01$ levels, our stabilized variance estimator preserves the type I error, especially in cases of large cluster size when it appears that the moment-based estimator is anti-conservative.

\subsection{A note on small sample size}
While data generated in Table \ref{table:typeIerror}a approach the asymptotic limit with regards to sample size ($n = 1000$), it is worth noting that this is rarely the case in real-world data analyses. This is significant, as when sample size is small, it is well documented that the robust sandwich covariance estimator for GEE will underestimate the true variance, thereby inflating the type I error. A variety of adjustments have been proposed to correct for this observed small sample bias \citep{fay2001small, skene2010analysis, mancl2001covariance, pan2002small}; consequently, it is straightforward to adapt our variance estimator to account for small sample bias by calculating $\bm{\hat{\Sigma}}_{k}$ using any one of the aforementioned variance adjustments. Table \ref{table:typeIerror}b illustrates this idea. Data was generated using the same underlying random intercepts model as in Table \ref{table:typeIerror}a and analyzed with an "independence" working correlation structure for MO-B, except the sample size was much smaller ($n = 10$). While the type I error of the stabilized variance estimator is inflated, it is still relatively close to the nominal level. When the bias correction proposed by Mancl and DeRouen is applied to our stabilized variance estimator, we find that the type I error actually falls below the nominal level. These results are not reported for the moment-based estimator, as we find that even with a large $M$, the variance is negative in a significant portion of the simulations.

\subsection{Power}
We also conducted 2,500 simulations under the alternative hypothesis $\beta \neq 0$ for $M = 2500$. We let $\sigma = (1, \sqrt{.8}, \sqrt{.5}, \sqrt{.1})$ and $\tau = (\sqrt{.1}, \sqrt{.5}, \sqrt{.8} , 1)$. We define $\rho$ as the within-cluster correlation, which is equal to $\frac{\tau^2}{\sigma^2 + \tau^2}$. Unlike in Table \ref{table:typeIerror}, we now choose an "exchangeable" working correlation for MO-B. Unsurprisingly, we find that when the data is analyzed with the correct working correlation structure, there are gains to statistical power when B increases. Moreover, the extent of these gains increases as the within-cluster correlation increases. These results validate the findings of \cite{chiang2008efficient}, which suggest that increasing B will result in a more asymptotically efficient variance estimator.

\begin{figure}[ht]
  \centering
  \resizebox{\linewidth}{!}{\input{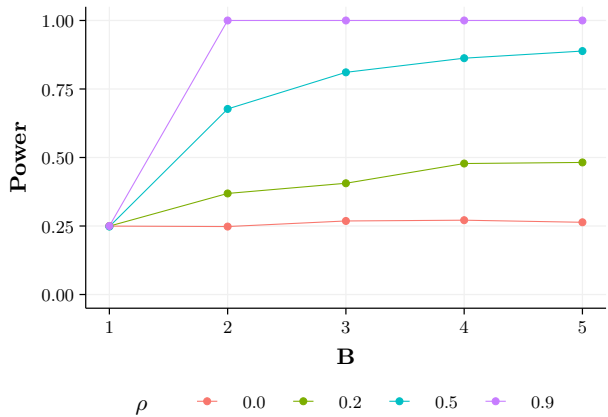}}
  \caption{Comparison of statistical power of MO and MO-B for $2 \leq B \leq 5$. Note that in this figure, $B = 1$ simply refers to the MO case. The data, which consists of $n = 1000$ subjects, was generated by varying the within-cluster correlation, $\rho$. For each condition, 2,500 simulations were conducted ($M = 2500$) under both the null ($\beta = 0$) and alternative hypothesis ($\beta \neq 0$) in order to calculate statistical power using our stabilized variance estimator. Moreover, for each value of $\rho$, an arbitrary critical value was chosen such that for $B = 1$, the power was roughly the same ($\approx0.25$), regardless of the within-cluster correlation.}
  \label{fig:power}
\end{figure}

\section{Real data application}
\begin{table*}[h]
\centering
\resizebox{.92\textwidth}{!}{%
\begin{tabular}{lllrrrrrr}
  \toprule
Response &  &  & Intercept & StdFlightLen & AvgFlightDur & ProbPause & MinsMissing & WkEndDayRtn \\ 
  \midrule
Anxiety & Coefficient & MO & 4.06E+00 & -2.59E-04 & -1.83E-04 & 1.63E+00 & -2.73E-03 & -1.83E+00 \\ 
   &  & MO-2 & 4.58E+00 & -1.77E-04 & -1.42E-04 & 1.36E+00 & -3.08E-03 & -1.59E+00 \\ 
   &  & MO-3 & 4.74E+00 & -1.48E-04 & -9.34E-05 & 1.28E+00 & -3.18E-03 & -1.52E+00 \\ 
   &  & MO-4 & 4.86E+00 & -1.31E-04 & -6.48E-05 & 1.24E+00 & -3.27E-03 & -1.48E+00 \\ 
   &  & MO-5 & 4.92E+00 & -1.19E-04 & -4.66E-05 & 1.22E+00 & -3.31E-03 & -1.46E+00 \\ 
   &  & GEE & 6.11E+00 & -5.53E-05 & 5.23E-07 & 7.21E-01 & -4.21E-03 & -8.47E-01 \\[.3cm]
   & Moment-based variance & MO & 3.72E+01 & 3.72E-07 & 1.58E-06 & 1.60E+00 & 1.93E-05 & 4.11E-01 \\ 
   &  & MO-2 & 1.02E+01 & -3.63E-08 & -2.03E-07 & -2.23E-01 & 4.91E-06 & 2.56E-01 \\ 
   &  & MO-3 & 1.21E+01 & -5.25E-09 & -4.40E-08 & 6.92E-02 & 5.90E-06 & 3.71E-01 \\ 
   &  & MO-4 & 1.28E+01 & 2.20E-09 & -1.60E-08 & 1.62E-01 & 6.25E-06 & 4.13E-01 \\ 
   &  & MO-5 & 1.32E+01 & 4.68E-09 & -6.50E-09 & 2.04E-01 & 6.43E-06 & 4.34E-01 \\ 
   &  & GEE & 1.25E+01 & 3.28E-09 & 1.11E-13 & 1.68E-01 & 5.82E-06 & 1.76E-01 \\[.3cm]
   & Stabilized variance & MO & 2.68E+00 & 4.78E-08 & 3.10E-07 & 2.53E-01 & 1.39E-06 & 6.04E-02 \\ 
   &  & MO-2 & 1.59E+00 & 1.05E-08 & 2.94E-08 & 7.58E-02 & 7.91E-07 & 4.58E-02 \\ 
   &  & MO-3 & 1.93E+00 & 7.76E-09 & 1.09E-08 & 7.20E-02 & 9.53E-07 & 5.99E-02 \\ 
   &  & MO-4 & 2.29E+00 & 6.40E-09 & 5.18E-09 & 7.48E-02 & 1.13E-06 & 7.38E-02 \\ 
   &  & MO-5 & 2.66E+00 & 5.62E-09 & 2.74E-09 & 7.99E-02 & 1.31E-06 & 8.78E-02 \\ 
   &  & GEE & -- & -- & -- & -- & -- & -- \\ [.3cm]
   & BC moment-based variance & MO & -- & -- & -- & -- & -- & -- \\
   & & MO-2 & 3.84E+01 & 4.74E-01 & 3.22E-06 & 2.37E+00 & 1.92E-05 & 1.43E+00 \\
   & & MO-3 & 2.97E+01 & 1.99E-01 & 9.29E-07 & 1.23E+00 & 1.47E-05 & 1.15E+00 \\
   & & MO-4 & 2.67E+01 & 1.12E-07 & 3.78E-07 & 8.86E-01 & 1.32E-05 & 1.04E+00 \\
   & & MO-5 & 2.52E+01 & 7.34E-08 & 1.90E-07 & 7.35E-01 & 1.25E-05 & 9.86E-01 \\
   &  & GEE & 2.04E+01 & 7.19E-09 & 1.74E-10 & 2.79E-01 & 9.52E-06 & 2.81E-01 \\ [.3cm]
   & BC stabilized variance & MO & -- & -- & -- & -- & -- & -- \\
   & & MO-2 & 3.34E+00 & 4.22E-08 & 2.40E-07 & 2.36E-01 & 1.68E-06 & 1.19E-01 \\
   & & MO-3 & 3.56E+00 & 2.66E-08 & 1.01E-07 & 1.80E-01 & 1.77E-06 & 1.32E-01 \\
   & & MO-4 & 4.01E+00 & 2.00E-08 & 5.38E-08 & 1.65E-01 & 1.99E-06 & 1.52E-01 \\
   & & MO-5 & 4.52E+00 & 1.63E-08 & 3.31E-08 & 1.62E-01 & 2.24E-06 & 1.73E-01 \\
   &  & GEE & -- & -- & -- & -- & -- & -- \\
   \bottomrule
\end{tabular}}
\caption{Coefficient and variance estimates for different methods with anxiety as the response variable. Multiple outputation (for all B) was performed with $M = 250,000$ outputations. GEE for MO-B and the full dataset was performed using an "independence" working correlation structure. Where applicable, the bias-correction proposed by Mancl and DeRouen was also employed. Note that even though M is large, the traditional moment-based approach still yielded negative variances. Similar analyses were conducted for the response variables of depression, sleep, and psychosis.}
\label{table:schizoutput}
\end{table*}
\cite{barnett2018relapse} first presented data obtained from seventeen patients being actively treated for schizophrenia at a state mental health clinic in Boston. Patients were followed daily for an average of 2 months, and mobility and social data (e.g. distance travelled, number of places visited, number of texts, number of calls) were passively recorded via the Beiwe app on each subject's personal smartphone. Subjects also self-reported their levels of anxiety, depression, psychosis, and sleep quality; in particular, depression as an outcome has already been shown by \cite{saeb2015mobile} to correlate with mobile phone sensor data.

Using this data, we compared our method of subsampling multiple observations per subject to traditional MO and GEE. Because measures of mobility and social data are often highly correlated with each other, we performed backward selection such that no predictor had a pairwise correlation with another predictor of greater than .5. Additionally, one subject was removed who only had 2 observations. From there, we performed MO and MO-B for $2 \leq B \leq 5$ with $M = 250,000$ using GEE with an "independence" working correlation structure. We also performed GEE on the entire dataset, similarly assuming an "independence" working correlation structure. In all analyses, the response variable is treated as continuous.

The results are shown in Table \ref{table:schizoutput} for anxiety as a response variable. As expected, we saw a close correspondence in the $\beta$'s for each of the various predictors among the different methods, with the exception of AvgFlightDur. Furthermore, though $M$ was quite large, using the moment-based variance estimate still occasionally led to negative variances, presumably due to the large cluster size compared to the small number of subjects. Moreover, assuming that the type I error is largely preserved, we observe that the use of our stabilized variance estimator often outperforms the moment-based estimator.

Given the small sample size, we also performed the analysis with the bias-correction (BC) proposed by Mancl and DeRouen for both the moment-based variance estimator and the stabilized variance estimator. As seen in Table \ref{table:schizoutput}, we find that the bias-correction inflates both the moment-based and stabilized variance estimates, which is as expected. Interestingly, we also see that the order of magnitude of the different estimates aligns closely with the simulation results; the BC moment-based variance estimate is consistently larger than the BC stabilized variance estimate, which may be reflective of the overly conservative type I error associated with the former (Table \ref{table:typeIerror}).

As a final note, we see that the stabilized variance associated with MO-2 is quite consistently smaller than the stabilized variance associated with MO, which could potentially suggest that subsampling more observations per person can increase performance. However, we also recognize that using MWCR without careful consideration of the setting can lead to certain issues; for instance, \citet{pavlou2013examination} describe how bias can arise in the MWCR method when the covariates are non-size balanced.

\section{Discussion}
In summary, we have developed a stabilized novel variance estimator for the WCR and MWCR methods. We have shown that in comparison to the moment-based estimator, our variance estimator 1) is strictly positive and 2) largely preserves the type I error. The fact that our variance estimator is non-negative facilitates the implementation of MO or MO-B, which can already be quite computationally intense. Moreover, the result that our variance estimator largely preserves the type I error is true even in the small-sample setting, in which existing proposed bias-corrections can be easily adapted to our method.

While our method certainly addresses a key limitation of the moment-based estimator, it is important to acknowledge that our method is still computationally intensive. To that note, \citet{chiang2008efficient} have proposed a second estimator, which is asymptotically equivalent to the moment-based estimator and avoids high computational cost. This estimator is analogous to the one associated with the cluster-weighted generalized estimating equation (CWGEE) approach proposed by \citet{williamson2003marginal}. Though we have not investigated this second estimator, in many scenarios, it might be preferable to our proposed variance estimator, simply because of computational cost.

Lastly, though we have applied our method to a real-world example and demonstrated improvement over the traditional moment-based approach, we recognize that in many settings, the outcome may not be continuous. Validating and extending these results to non-continuous (e.g. binary) settings may be of interest in the future.

\backmatter


\newpage


%

\appendix


\label{lastpage}

\end{document}